\documentclass[aps,showpacs,prl,twocolumn,superscriptaddress,floatfix,longbibliography]{revtex4-1}

\usepackage{amsthm}
\usepackage[utf8]{inputenc}
\usepackage{amsmath,amssymb,amsmath,amsthm,bm,graphicx,xcolor}
\usepackage{natbib}
\usepackage[breaklinks]{hyperref}
\usepackage{color}
\usepackage{tabularx}
\usepackage{epsfig}
\usepackage{amssymb}
\usepackage{graphicx}
\usepackage{wasysym}
\usepackage{dcolumn}
\usepackage{epstopdf}
\usepackage{subfigure}
\usepackage{psfrag}
\usepackage{bbold}
\usepackage[breaklinks]{hyperref}
\hypersetup{colorlinks=true}

\newcommand{\vecform}{\bm}

\newcommand{\bra}[1]{\mbox{$\langle #1 |$}}

\newcommand{\ket}[1]{\mbox{$| #1 \rangle$}}

\newcommand{\Tr}{\rm Tr}    

\newtheorem{thm}{Theorem}

\begin{document}

\title{A representability problem in density functional theory for superconductors}

\author{Jonathan Schmidt}
\affiliation{Institut f\"ur Physik, Martin-Luther-Universit\"at
Halle-Wittenberg, 06120 Halle (Saale), Germany}
  
\author{Carlos L. Benavides-Riveros}
\email{carlos.benavides-riveros@physik.uni-halle.de}
\affiliation{Institut f\"ur Physik, Martin-Luther-Universit\"at
Halle-Wittenberg, 06120 Halle (Saale), Germany}

\author{Miguel A. L. Marques}
\affiliation{Institut f\"ur Physik, Martin-Luther-Universit\"at
Halle-Wittenberg, 06120 Halle (Saale), Germany}

\date{\today}

\begin{abstract}
Aiming at a unified treatment of correlation 
and inhomogeneity effects in superconductors, 
Oli\-veira, Gross and Kohn proposed in 1988
a density functional theory for the superconducting 
state. This theory relies on the existence of a Kohn-
Sham scheme, i.e., an auxiliary noninteracting system 
with the same electron and anomalous densities
of the original superconducting system. However, the 
question of noninteracting $v$-representability has 
never been properly addressed and the existence of the 
Kohn-Sham system has always been assumed
without proof. Here, we show that indeed such a 
noninteracting system does not exist in at zero 
temperature. In spite of this result, we
also show that the theory is still able to yield 
good results, although in the limit of weakly correlated 
systems only.
\end{abstract}

\maketitle

The correct description of superconductors continues to be 
one of the great questions of modern condensed matter theory. 
We still have a poor understanding of the superconducting 
features of high-$T_c$ cuprates \cite{RevModPhys.78.17} or layered 
organic materials \cite{0034-4885-74-5-056501}. Moreover, recent 
exciting discoveries of superconductivity in two layers of graphene 
rotated by a small angle \cite{graphene,PhysRevLett.121.087001}, 
or in hydrides at high pressure~\cite{Eremets,PhysRevLett.114.157004} 
(whose transition temperatures are quickly approaching room 
temperature) continue to challenge our understanding of 
superconducting systems.

Until the discovery of unconventional superconductivity by Bednorz and
M\"uller \cite{bednorz1986possible}, Bardeen-Cooper-Schiefer (BCS)
theory \cite{PhysRev.108.1175} as well as its theoretical extensions
\cite{osti_7354388,gorkov} were the main framework for understanding
superconductivity. While the electron-phonon interaction is well
accounted for in BCS and Eliashberg's theories, correlation effects
due to the electron-electron Coulomb repulsion are extraordinarily
difficult to handle. Since those effects are usually condensed in a
single parameter (namely, $\mu^*$), which is almost always treated as an
adjustable quantity, the predictive capability of the theory is quite
limited.

Aiming for a unified treatment of electronic correlation and
inhomogeneity effects in superconductors, in 1988 Oliveira, Gross and
Kohn proposed a density functional theory (DFT) for the
superconducting state (SCDFT) \cite{PhysRevLett.60.2430}.  Normal DFT
is based on the famous Hohenberg-Kohn theorems \cite{HK}, whose main
statement is the observation that for non-degenerate systems there is
a one-to-one correspondence between the external potential, the
(non-degenerate) many-body ground state and the associated
ground-state electronic density. The DFT formalism for
superconductivity is based on a similar theorem, the
Oliveira-Gross-Kohn theorem \cite{PhysRevLett.60.2430}, which
guarantees a one-to-one mapping between the equilibrium statistical
density operator and the electronic $n(r) \equiv \langle
\hat\psi^\dagger(r) \hat\psi(r)\rangle$ and the anomalous
$\chi(r,r')\equiv \langle \hat\psi_{\uparrow}(r)
\hat\psi_{\downarrow}(r')\rangle$ densities:
\begin{align}
(n,\chi) \overset{\tiny {\rm 1-1}}{\Longleftrightarrow}  \hat \rho,
\end{align}
where $\hat \rho = e^{-\beta \hat{H}_{v\Delta}}/ {\Tr}[e^{-\beta
    \hat{H}_{v\Delta}}]$ is the equilibrium statistical operator,
    and $\hat{H}_{v\Delta}$ is the grand-canonical Hamiltonian for a
superconductor in an external potential $v(r)$ and a non-local pair
potential $\Delta(r,r')$. The field operator $\hat\psi_\sigma(r)$
($\hat\psi^\dagger_\sigma(r)$) destroys (creates) an electron with
spin $\sigma$ at position $r$.

Although SCDFT was later extended to consider the nuclear density
\cite{SCDFT1,SCDFT2}, the electronic problem is still trea\-ted at the
level of a two-component DFT. This version of SCDFT has been extremely
successful in predicting superconductivity in a wide variety of
materials \cite{cac6scdft,dissilanescdftpressure,PBSCDFT} and proved
especially useful for the investigation of superconductivity in
high-pressure environments
\cite{metalsscdftpressure,sulfurhybridesscdftpressure,seleniumscdftpressure}.
This success of SCDFT stems mainly from two facts: First, the gap
equation is a single 3-dimensional integral equation, while Eliashberg
equations form a system of 4-dimensional integral equations in
momentum and frequency space; Second, it is much simpler to
approximate and evaluate the electron-electron interaction term in
SCDFT. This allows us not to use a semi-empirical $\mu^*$, rendering
therefore SCDFT a truly \textit{ab initio}, predictive theory.

In principle, the formalism is able to describe superconductivity in 
all systems, for all complexities of the many-electron problem are cast 
into a universal exchange-co\-rre\-lation functional whose existence is 
ensured  by the  Oliveira-Gross-Kohn theorem. Yet, the practical 
applicability 
of SCDFT rests (as in the standard DFT) on a Kohn-Sham scheme, which 
maps the real interacting system of interest to an auxiliary  
noninteracting one with the same equilibrium electronic and 
anomalous densities. The Kohn-Sham system of SCDFT is described 
in second quantization by the following Hamiltonian:
\begin{align}
\nonumber
H_\text{s} &= -\sum_\sigma \int \hat\psi_\sigma^{\dagger}(r)\left[\frac{\nabla^2}{2}+\mu-
v_\text{s}(r)\right]\hat\psi_\sigma(r)dr\nonumber\\ &\qquad
-\int\left[ \Delta_\text{s}^*(r,r')\hat\psi_\uparrow(r)\hat\psi_\downarrow(r')+h.c.\right]drdr',
\label{eq:hamiltonian}
\end{align}
where $\mu$ is the chemical potential and $v_\text{s}(r)$ and
$\Delta_\text{s}(r,r')$ are the normal and anomalous Kohn-Sham 
potentials, computed by means of functional derivatives of the 
universal SCDFT functional.

In standard DFT, the question whether one can always find a
noninteracting system with external potential $v_\text{s}$ for which
the ground-state electron density is the same of the full interacting
system, is known as the \textit{noninteracting $v$-representability}
problem \cite{Engel,VANLEEUWEN200325,10.1007/978-3-319-04912-0_2}.  In
SCDFT a system is noninteracting ($v, \Delta$)-representable if a
noninteracting system defined through the potentials $v_\text{s}(r)$
and $\Delta_\text{s}(r,r')$ can be found such that the equilibrium
densities of the system are respectively equal to $n(r)$ and
$\chi(r,r')$. Unlike zero-temperature DFT~\cite{Levy6062}, in SCDFT this
question has never been properly addressed and the existence of the
Kohn-Sham system has always been assumed without proof. In this letter
we prove that indeed such a noninteracting system does 
\textit{not exist}, at least at zero temperature!

By and large, this is a difficult problem in the light of the 
well-known fact that the set of interacting pure-state $v$-representable 
electronic densities 
\begin{align}
\mathcal{B} &= \left\{ n = \bra{\Psi}\hat n\ket{\Psi} \,|\,
\ket{\Psi} \in \mathcal{G} \right\},
\end{align}
where $\mathcal{G}$ is the set of ground states
for some external potential $v$, is not identical to
the set of noninteracting pure-state $v$-representable 
electronic densities \cite{doi:10.1002/qua.560240302}.
As such, representability questions are well-known problems 
in many-body physics with no general answers. 
For instance, for spin-polarized systems two different 
non-degenerate ground states always lead to two different 
sets of electronic and magnetization densities. Yet two different 
sets of external potentials $(v, \vecform{B})$ can lead to the 
same ground state, and therefore ---unlike standard DFT--- 
the potentials of spin-polarized systems are not unique 
functionals of the spin densities \cite{PhysRevLett.86.5546}.

A second example, that will turn out to be essential for our 
purposes, is reduced density matrix functional theory (RDMFT). 
Instead of the density, in RDMFT the main object is $\hat \rho_1$, 
the one-particle reduced density matrix, whose diagonal is 
the electronic density $n(r)$. Similarly to DFT, RDMFT is 
based on a variational principle stating that the ground-state 
energy of a fermionic system can be obtained by minimizing some 
energy functional on the set of $N$-representable 
one-body reduced density matrices \cite{PhysRevB.12.2111}.
The condition that $\hat\rho_1$ must satisfy in order to ensure its 
$N$-representability (i.e.~that there is a fermionic 
statistical operator whose contraction leads to $\hat\rho_1$) 
reads simply \cite{Col2}: $0 \leq \hat \rho_1 \leq 1$, provided 
that $ {\Tr} [\hat \rho_1] = N$. For pure systems these 
representability conditions are known to be more stringent \cite{Kly2,SBV}.
Yet, when the fermionic density operator corresponds to a 
noninteracting system at zero temperature (and the many-body 
state reduces to a single Slater determinant) the one-body 
reduced density matrix is in addition idempotent, namely, 
\begin{align} 
\label{eq:idem}
\hat \rho_1^2 = \hat\rho_1.
\end{align}
Zero-temperature noninteracting one-body reduced densities are on the
boundary of the set of all fermionic one-body reduced densities. 
Since for interacting systems $\hat \rho_1$ is not idempotent 
(in fact, $\hat \rho_1^2 < \hat\rho_1$~\cite{Seir}) it is known 
that there is no Kohn-Sham system for RDMFT at zero temperature.
In other words, interacting one-body reduced density matrices are non-\textit{non\-interacting $v$-representable} at zero temperature. 

Now, returning to the problem of noninteracting 
re\-pre\-sen\-tability in SCDFT, we notice that the Hamiltonian
\eqref{eq:hamiltonian} can be diagonalized by a Bogoliubov
transformation of the field operators $\hat \psi_\sigma$:
\begin{align}
\label{eq:BdGtrnasf}
\hat\psi_{\sigma}(r) = 
\sum_{k} u_{k}(x) \hat\gamma_{k,\sigma}-\text{sign}(\sigma)v^*_{k}(x) 
\hat\gamma^{\dagger}_{k,-\sigma}.
\end{align}
Here $\hat\gamma^\dagger_{k,\sigma}$ and $\hat\gamma_{k,\sigma}$ 
are creation and annihilation operators of fermionic quasiparticles.
The transformation \eqref{eq:BdGtrnasf} leads to a set of 
Bogo\-liu\-bov-de Gennes equations for the coefficient 
functions $u_k(r)$ and $v_k(r)$ as well as the excitation 
energies of the quasiparticles $E_k$. These equations are 
the counterpart to the Kohn-Sham orbital equations in normal DFT. 
Moreover, we arrive at a set of self-consistent equations 
for the ground-state densities, which take the following form at 
zero temperature:
\begin{align}
\nonumber
\rho_1(r,r')&=\sum_{k} u_k(r) u_k^*(r') \theta(-E_k) + v_k(r)v_k^*(r')\theta(E_k), \\
\chi(r,r')&=\sum_{k} u_{k}(r)v^*_{k}(r')\theta(E_k)-u_{k}(r')v^*_{k}(x)\theta(-E_k),
\end{align}
with $\theta(x)$ being the step function. Clearly, the ground-state 
density is equal to $n(r)=\rho_1(r,r)$. In order to treat the problem 
of the existence of the Kohn-Sham system in the zero-temperature case, 
we will make use of the Nambu-Gorkov formalism. In Nambu-Gorkov space 
the field operators are defined as the following spinors:
\begin{align}
\bar{\psi}_{\downarrow}(r) =
\bigg(\begin{array}{c} \hat\psi^{\dagger}_{\downarrow}(r) \\ 
\hat\psi_{\uparrow}(r) \end{array}\bigg)
\quad {\rm and} \quad
\bar{\psi}_{\uparrow}(r)=
\bigg(\begin{array}{c} \hat\psi^{\dagger}_{\uparrow}(r) \\ 
\hat\psi_{\downarrow}(r) \end{array}\bigg).
\end{align}
We can write the generalized one-body reduced density matrix as a 
tensor-product of Nambu-Gorkov field-operators:
\begin{align}
    \Gamma(r,r')= \langle\bar{\psi}_{\downarrow}(r)
    \otimes \bar{\psi}^\dagger_{\downarrow}(r')\rangle .
    \end{align}
Obviously, $\Gamma(r,r')$ can be expressed in terms of the normal 
one-body reduced density matrix $\rho_1(r,r')$ and the anomalous 
density $\chi(r,r')$.

To answer the question of \textit{noninteracting 
$(v,\Delta)$-re\-pre\-sen\-ta\-bi\-lity} of the superconducting 
densities it is worth studying the structure of the ground states 
of the Hamiltonian \eqref{eq:hamiltonian}. In the case of the 
normal noninteracting electronic system (i.e.~$\Delta_\text{s} = 0$)
the ground-state one-body reduced density matrix is idempotent 
\eqref{eq:idem}. Surprisingly, the same is true for the 
Nambu-Gorkov generalized one-body reduced density matrix for the 
noninteracting Hamiltonian~\eqref{eq:hamiltonian}, namely, 
$\hat \Gamma^2 = \hat \Gamma$. It is simple to show that 
idempotency of $\hat \Gamma$ reduces to the equation:
\begin{align}
\label{Eq.initial}
\hat \rho_1 = \hat \rho_1^2 + \chi^\dagger\chi.
\end{align}
This is a particular case of a general result in mathematical physics,
especially generalized Hartree-Fock-Bogoliubov theory \cite{Bach1994}. 
Indeed, by a generalization of the Lieb’s variational principle,
for systems with semi-bounded Hamiltonians, the infimum of the energy 
over quasifree states (the ones satisfying Wick's theorem) is reached 
by a pure state \cite{doi:10.1063/1.4853875}. A quasifree 
state is \textit{pure} if and only if $\hat \Gamma$ is idempotent 
\cite{doi:10.1063/1.4853875,Bach1994}.

Yet in SCDFT we are not interested in $\rho_1(r,r')$ but in the 
electronic density $n(r)$. From Eq.~\eqref{Eq.initial}, for a 
superconducting noninteracting system, $n(r)$ can be written in 
terms of $\hat\rho_1$ and $\chi$. It reads:
\begin{align}\label{eq:10}
n(r)= {\Tr} \left[\hat\rho^2_1(r,r')+|\chi(r,r')|^2\right].
\end{align}
By construction, in SCDFT no prior relation exists between $n(r)$ and $\chi(r,r')$ and they are assumed to be
independent variables. 

We will now demonstrate that in the Kohn-Sham system these two 
quantities are not independent for translation-invariant 
systems, for which the one-body reduced density matrix  
$\rho_1(r,r') = \rho_1(r-r',0)$ only depends on the differences $r-r'$. 
In reciprocal space such density depends on one variable:
\begin{align}
\rho_1(k) = \int \rho_1(r-r',0)e^{-ik(r-r')} d(r-r').
\end{align}
Analogously we can write $\chi(r,r')$ as $\chi(k)$ \cite{Hainzl2008}.
For the Kohn-Sham noninteracting system $\Gamma^2 = \Gamma$, and 
therefore 
$\rho_1(k)=\rho_1^2(k)+|\chi(k)|^2$. It yields:
\begin{align}
\rho_1(k)=\frac{1 + \text{sign}(E_k) \sqrt[]{1-4|\chi(k)|^2}}{2}.
\end{align} 
We already noted that Eq.~\eqref{eq:10} is valid in the 
noninteracting system case. For translation-invariant systems 
this results in:
\begin{gather}
n(r)=\int \frac{1}{2} \left(1+\text{sign}(E_k)
\sqrt[]{1-4|\chi(k)|^2}\right)dk
\end{gather}
In other words, we have an explicit equation for the electron 
density in terms of $\chi(r,r')$ and we can write:
\begin{gather}
\frac{\delta n(r)}{\delta \chi(k)}=-\frac{2\,\text{sign}(E_k)}
{\sqrt{1-4|\chi(k)|^2}}\chi^*(k),
\end{gather}
As expected,  $n$ and $\chi(r,r')$ are \textit{not} independent 
and one is univocally determined by each other. 
Now, it is evident that the set of noninteracting 
$(v,\Delta)$-representable 
densities and the set of interacting $(v,\Delta)$-representable 
densities cannot coincide. Hence, no Kohn-Sham system can in general 
exist at zero temperature in the formalism of SCDFT.  

It is possible to gain further insight into this problem at least 
for weakly correlated systems. Indeed, it is quite remarkable that 
for the homogeneous electron gas including electron-electron 
interactions to the Hamiltonian \eqref{eq:hamiltonian} maintains 
the idempotency of $\hat \Gamma$ in first-order perturbation theory. 
To see that let us perturb $H_\text{s}$ with the Coulomb interaction, 
namely,  $\hat H = \hat H_\text{s} + \lambda \hat W_{\rm e-e}$. In 
first-order perturbation theory the perturbed wave function includes 
only double and quadruple excitations of the quasiparticle vacuum 
(the superconducting ground state) $\ket{0}$ (see Supplemental 
Material). It reads:
\begin{align}
\label{eq:mod}
\ket{0_\lambda} &= \ket{0} + \lambda \sum_k b_k 
\hat \gamma^\dagger_{k0} \hat\gamma^\dagger_{k1}\ket{0} 
\\ & \qquad + \lambda \sum_{qkk'} d_{qkk'} 
\hat\gamma^\dagger_{(k+q)0} \hat\gamma^\dagger_{-(k'-q)1} 
\hat\gamma^\dagger_{-k'0} \hat\gamma^\dagger_{k1}\ket{0} 
+ \mathcal{O}(\lambda^2), \nonumber
\end{align}
where the amplitudes are given as the expected value of the 
interelectronic interaction between the quasiparticle vacuum and 
the excited state $\hat \gamma^\dagger_{k0} \hat\gamma^\dagger_{k1}
\ket{0}$ (with energy $E_k$):
\begin{align}
b_k = \frac{\bra{0} \hat \gamma_{k1} \hat\gamma_{k0}
\hat W_{\rm e-e}\ket{0}} {E_0 - E_k},
\end{align}
where $E_0$ is the vacuum's energy, as well as with the 
state $\hat\gamma^\dagger_{(k+q)0}\hat\gamma^\dagger_{-(k'-q)1} 
\hat\gamma^\dagger_{-k'0} \hat\gamma^\dagger_{k1}\ket{0}$ 
(with energy $E_{q kk'}$),
\begin{align}
d_{qkk'} = \frac{\bra{0} \hat\gamma_{k1}\hat\gamma_{-k'0}
\hat\gamma_{-(k'-q)1} \hat\gamma_{(k+q)0}
 \hat W_{\rm e-e} \ket{0}}{E_0 - E_{q kk'}}.
\end{align}
Notice that in Eq.~\eqref{eq:mod} 
$\bra{0_\lambda}0_\lambda\rangle = 1 + \mathcal{O}(\lambda^2)$, 
so that in leading order the wave function is fairly normalized. 
It is easy to see that the perturbed diagonal densities read
\begin{align*}
\rho'_{k} = \rho_{k} + \lambda \rho^{(1)}_{k} 
+ \mathcal{O}(\lambda^2), \; 
\; \chi'_{k} = \chi_{k} + \lambda \chi^{(1)}_{k} 
+ \mathcal{O}(\lambda^2),
\end{align*} 
where $\rho^{(1)}_{k} = b_k u_k^*v_k^* + b_k^* u_k v_k$ and 
$\chi^{(1)}_{k} = -b_k u_k^{*2} + b^*_k v_k^2$.
Since $2 \rho_{k} \rho^{(1)}_{k} + (\chi^*_{k}\chi_{k}^{(1)} + 
\text{c.c.}) 
= \rho_{k}^{(1)}$, $\hat \Gamma$ is idempotent in first order 
perturbation theory.

This result indicates that, for weakly correlated systems, in
principle it is still possible to find a noninteracting system that
reproduces the corresponding densities in first order of the
perturbation. This result also explains why despite the fact that a
Kohn-Sham system does not exist in general in SCDFT, the theory is
still able to give good results for weakly correlated systems. A
different conclusion arises from second-order perturbation theory,
since the idempotency is, indeed, broken at such order (see Supplemental
Material).
 
In layman's terms, at zero temperature the existence of 
a Kohn-Sham system is equivalent to the fact that for 
every interacting pair $(n,\chi)$ there is a noninteracting 
system that has the same noninteracting $(n,\chi)$. However, 
our result implies that there is a class of systems where all 
noninteracting densities are of the form 
$(n,\chi[n])$.
This immediately proves that there are an infinite set of interacting 
$(n,\chi)$ that are not noninteracting $(v,\Delta)$-representable,
for which therefore no Kohn-Sham system exists.

We note that this problem of non-interacting 
$(v,\Delta)$-re\-pre\-sen\-tability is to some extent 
caused by the use of both a local density $n$ and a component 
of the density matrix $\chi$. As such, DFT for superconductors 
is somewhat a hybrid theory, inheriting the problems of both DFT 
and RDMFT. A possible workaround is to simply use the density matrix in
Nambu-Gorkov space as a fundamental variable. This would lead to a
more symmetric and elegant theory, namely a reduced-density matrix
theory for superconductors, that would circumvent many of the problems
of the current approach. Work along this lines is already in pro\-gress.

We thank Hardy Gross for helpful discussions. We acknowledge financial
support from the DFG through Projects No. SFB-762 and No. MA 6787/1-1
(M.A.L.M.).

\bibliography{Laverna}

\onecolumngrid

\newpage

\appendix 

\section{SUPPLEMENTAL MATERIAL}

In this Supplemental Material we provide the proof of two 
statements made in the main text. First, that for the 
homogeneous electron gas the generalized one-particle reduced 
density operator $\Gamma(\lambda)$ is idempotent in first order 
perturbation theory. Since the infimum of the energy for the 
Bogo\-liu\-bov-de Gennes Hamiltonian is reached by a pure state, 
with an idempotent generalized one-particle reduced density matrix, 
this can explain why SCDFT yields good results for many 
superconductors. Second, that $\Gamma(\lambda)$ is not 
idempotent in second order. 

The Hamiltonian we consider here is $\hat H = \hat H_\text{s} + \lambda 
\hat W_{\rm e-e}$, where $\hat H_\text{s}$ is the noninteracting Bogo\-liu\-bov-de 
Gennes Ha\-mil\-tonian introduced in Eq.~\eqref{eq:hamiltonian} and 
$\hat W_{\rm e-e}$ is the electronic repulsion operator.
Let us use the well-known Bogoliubov transformations for the electron 
creation and annihilation operators ($\hat c_{k\sigma}$,$\hat 
c^\dagger_{k\sigma}$), namely,
\begin{align}
\hat c_{k\uparrow}
 &= u^*_k \hat \gamma_{k0} + v_k \hat \gamma^\dagger_{k1}, \\
\hat c_{-k\downarrow}^\dagger &= -v^*_k \hat \gamma_{k0} + u_k \hat 
\gamma^\dagger_{k1} .
\end{align}
The multiplication of creation and annihilation electron 
operators reads in the quasiparticle operators:
\begin{align}
\label{eq:ap1}
\hat c^\dagger_{k\uparrow} \hat c_{k\uparrow} &= |u_k|^2 
\hat \gamma^\dagger_{k0}\hat \gamma_{k0}
+ u_k v_k \hat \gamma^\dagger_{k0}\hat \gamma^\dagger_{k1} +
u_k^* v_k^* \hat\gamma_{k1}\hat\gamma_{k0}
+ |v_k|^2 \hat\gamma_{k1}\hat\gamma_{k1}^\dagger \\
\hat c_{k\uparrow} \hat c_{-k\downarrow} &= -u^*_k v_k \hat\gamma_{k0}
\hat\gamma^\dagger_{k0} + u_k^{*2} \hat\gamma_{k0}\hat\gamma_{k1} - v_k^{2} 
\hat\gamma^\dagger_{k1}\hat\gamma^\dagger_{k0}
+ v_k u^*_k \hat\gamma^\dagger_{k1}\hat\gamma_{k1} .
\label{eq:ap2}
\end{align}
At zero order (i.e, $\lambda = 0$) the ground state is $\ket{\rm BCS}$,
the quasiparticle vacuum or more commonly the BCS state. 
The densities at zero order are given by 
\begin{align}
\rho^{(0)}_{kk} &\equiv \bra{\rm BCS}\hat c^\dagger_{k\uparrow} \hat c_{k\uparrow}\ket{\rm BCS}
= |v_k|^2 \\
\chi^{(0)}_{kk} &\equiv \bra{\rm BCS}\hat c_{k\uparrow} \hat c_{-k\downarrow}\ket{\rm BCS}
= -u_k^* v_k
\end{align}
Notice that
\begin{align}
(\rho^{(0)}_{kk})^2 + 
|\chi^{(0)}_{kk}|^2 
= |v_k|^4 + |v_k|^2 |u_k|^2 = \rho_{kk}^0
\end{align}
and trivially $\rho^{(0)}_{kk} \chi^{(0)}_{kk} + (1-\rho^{(0)}_{kk}) \chi^{(0)}_{kk} 
= \chi^{(0)}_{kk}$. These are nothing more than the conditions for 
the idempotency at zero order, a quite well known result in superconductivity.

The particular combination of field-operators that enter the Coulomb interaction terms and produce relevant terms for the 1RDM/anomalous density is
\begin{align}
\label{eq:eeint}
& \hat c^\dagger_{(k+q)\uparrow} \hat c^\dagger_{(k'-q)\downarrow}
\hat c_{k'\downarrow}\hat c_{k\uparrow} =
\\ &\qquad(u_{k+q}\hat\gamma^\dagger_{(k+q)0} + v^*_{k+q}\hat\gamma_{(k+q)1})
(-v^*_{-(k'-q)}\hat\gamma_{-(k'-q)0} + u_{-(k'-q)}\hat\gamma^\dagger_{-(k'-q)1})
(-v_{-k'}\hat\gamma^\dagger_{-k'0} + u^*_{-k'}\hat\gamma_{-k'1}) 
(u^*_k\hat\gamma_{k0} + v_k\hat\gamma^\dagger_{k1}) ,
\nonumber
\end{align}
and expanding the products yields
\begin{multline}
\hat c^\dagger_{(k+q)\uparrow} \hat c^\dagger_{(k'-q)\downarrow}
\hat c_{k'\downarrow}\hat c_{k\uparrow}
= \\
+ u_{k+q} v^*_{-(k'-q)}v_{-k'} v_k
\hat\gamma^\dagger_{(k+q)0}\gamma_{-(k'-q)0}\hat\gamma^\dagger_{-k'0}\hat\gamma^\dagger_{k1}
- u_{k+q} v^*_{-(k'-q)}u^*_{-k'} v_k
\hat\gamma^\dagger_{(k+q)0}\gamma_{-(k'-q)0}\hat\gamma_{-k'1}\hat\gamma^\dagger_{k1}
\\
- u_{k+q} u_{-(k'-q)}v_{-k'} v_k
\hat\gamma^\dagger_{(k+q)0}\gamma^\dagger_{-(k'-q)1}\hat\gamma^\dagger_{-k'0}\hat\gamma^\dagger_{k1}
+ u_{k+q} u_{-(k'-q)}u^*_{-k'} v_k
\hat\gamma^\dagger_{(k+q)0}\gamma^\dagger_{-(k'-q)1}\hat\gamma_{-k'1}\hat\gamma^\dagger_{k1}
\\
+  v^*_{k+q} v^*_{-(k'-q)}v_{-k'} v_k 
\hat\gamma_{(k+q)1}\gamma_{-(k'-q)0}\hat\gamma^\dagger_{-k'0}\hat\gamma^\dagger_{k1}
- v^*_{k+q} v^*_{-(k'-q)}u^*_{-k'} v_k
\hat\gamma_{(k+q)1}\gamma_{-(k'-q)0}\hat\gamma_{-k'1}\hat\gamma^\dagger_{k1}
\\
- v^*_{k+q} u_{-(k'-q)}v_{-k'} v_k
\hat\gamma_{(k+q)1}\gamma^\dagger_{-(k'-q)1}\hat\gamma^\dagger_{-k'0}\hat\gamma^\dagger_{k1}
+ v^*_{k+q} u_{-(k'-q)}u^*_{-k'} v_k
\hat\gamma_{(k+q)1}\gamma^\dagger_{-(k'-q)1}\hat\gamma_{-k'1}\hat\gamma^\dagger_{k1}
\\
+ u_{k+q} v^*_{-(k'-q)}v_{-k'} u^*_k
\hat\gamma^\dagger_{(k+q)0}\gamma_{-(k'-q)0}\hat\gamma^\dagger_{-k'0}\hat\gamma_{k0}
- u_{k+q} v^*_{-(k'-q)}u^*_{-k'} u^*_k
\hat\gamma^\dagger_{(k+q)0}\gamma_{-(k'-q)0}\hat\gamma_{-k'1}\hat\gamma_{k0} 
\\
- u_{k+q} u_{-(k'-q)}v_{-k'} u^*_k
\hat\gamma^\dagger_{(k+q)0}\gamma^\dagger_{-(k'-q)1}\hat\gamma^\dagger_{-k'0}\hat\gamma_{k0}
+ u_{k+q} u_{-(k'-q)}u^*_{-k'} u^*_k
\hat\gamma^\dagger_{(k+q)0}\gamma^\dagger_{-(k'-q)1}\hat\gamma_{-k'1}\hat\gamma_{k0}
\\
+  v^*_{k+q} v^*_{-(k'-q)}v_{-k'} u^*_k
\hat\gamma_{(k+q)1}\gamma_{-(k'-q)0}\hat\gamma^\dagger_{-k'0}\hat\gamma_{k0}
- v^*_{k+q} v^*_{-(k'-q)}u^*_{-k'} u^*_k
\hat\gamma_{(k+q)1}\gamma_{-(k'-q)0}\hat\gamma_{-k'1}\hat\gamma_{k0}
\\
- v^*_{k+q} u_{-(k'-q)}v_{-k'} u^*_k
\hat\gamma_{(k+q)1}\gamma^\dagger_{-(k'-q)1}\hat\gamma^\dagger_{-k'0}\hat\gamma_{k0}
+ v^*_{k+q} u_{-(k'-q)}u^*_{-k'} u^*_k
\hat\gamma_{(k+q)1}\gamma^\dagger_{-(k'-q)1}\hat\gamma_{-k'1}\hat\gamma_{k0}.
\end{multline}

We are interested in the expectation value between the BCS state 
and the excited states of the Bogo\-liu\-bov-de Gennes Hamiltonian 
(say, $\ket{s}$). We have therefore
\begin{multline}
\bra{s}\hat c^\dagger_{(k+q)\uparrow} \hat 
c^\dagger_{(k'-q)\downarrow} 
\hat c_{k'\downarrow}\hat c_{k\uparrow}\ket{\rm BCS} = \\
u_{k+q} v^*_{-(k'-q)}v_{-k'} v_k \bra{s}
\hat\gamma^\dagger_{(k+q)0}\gamma_{-(k'-q)0}\hat\gamma^\dagger_{-k'0}\hat\gamma^\dagger_{k1}
\ket{\rm BCS}
- u_{k+q} u_{-(k'-q)}v_{-k'} v_k
\bra{s}\hat\gamma^\dagger_{(k+q)0}\gamma^\dagger_{-(k'-q)1}\hat\gamma^\dagger_{-k'0}\hat\gamma^\dagger_{k1}\ket{\rm BCS}
\\
+u_{k+q} u_{-(k'-q)}u^*_{-k'} v_k
\bra{s}
\gamma^\dagger_{(k+q)0}\gamma^\dagger_{-(k'-q)1}\hat\gamma_{-k'1}\hat\gamma^\dagger_{k1}
\ket{\rm BCS}
- v^*_{k+q} u_{-(k'-q)}v_{-k'} v_k
\bra{s}\hat\gamma_{(k+q)1}\gamma^\dagger_{-(k'-q)1}\hat\gamma^\dagger_{-k'0}
\hat\gamma^\dagger_{k1}\ket{\rm BCS}
\end{multline}
that reduces to
\begin{multline}
\bra{s}\hat c^\dagger_{(k+q)\uparrow} \hat c^\dagger_{(k'-q)\downarrow} 
\hat c_{k'\downarrow}\hat c_{k\uparrow}\ket{\rm BCS}
= \delta^0_q
u_{k} |v_{-k'}|^2 v_k
\bra{s}
\hat\gamma^\dagger_{k0}\hat\gamma^\dagger_{k1}\ket{\rm BCS}
+ \delta_{-k'}^k
u_{k+q} u_{k+q}u^*_{k} v_k
\bra{s}
\gamma^\dagger_{(k+q)0}\gamma^\dagger_{(k+q)1}\ket{\rm BCS}
\\
- \delta_{-k'}^k v^*_{k+q} u_{k+q}v_{k} v_k
\bra{s}\hat\gamma^\dagger_{k0}\hat\gamma^\dagger_{k1}\ket{\rm BCS}
+ \delta^0_q|v_{k}|^2 u_{-k'}v_{-k'} 
\bra{s}\gamma^\dagger_{-k'1}\hat\gamma^\dagger_{-k'0}\ket{\rm BCS}
\\ 
- u_{k+q} u_{-(k'-q)}v_{-k'} v_k
\bra{s}\hat\gamma^\dagger_{(k+q)0}\gamma^\dagger_{-(k'-q)1}\hat\gamma^\dagger_{-k'0}\hat\gamma^\dagger_{k1}\ket{\rm BCS}.
\end{multline}
For the homogeneous electron gas we only have to consider $q \neq
0$ which reduces the contributions to solely:
\begin{align}
\label{eq:red}
&\bra{s}c^\dagger_{(k+q)\uparrow} c^\dagger_{(k'-q)\downarrow}c_{k'\downarrow}c_{k\uparrow}
\ket{\rm BCS}
\\ &= \delta_{-k'}^k (u_{k+q} u_{k+q}u^*_{k} v_k
\bra{s}
\gamma^\dagger_{(k+q)0}\gamma^\dagger_{(k+q)1}\ket{\rm BCS} \nonumber
- v^*_{k+q} u_{k+q}v_{k} v_k
\bra{s}\hat\gamma^\dagger_{k0}\hat\gamma^\dagger_{k1}\ket{\rm BCS})
\nonumber 
\\ 
&- u_{k+q} u_{-(k'-q)}v_{-k'} v_k
\bra{s}\hat\gamma^\dagger_{(k+q)0}\gamma^\dagger_{-(k'-q)1}\hat\gamma^\dagger_{-k'0}\hat\gamma^\dagger_{k1}\ket{\rm BCS} \nonumber 
\\ &= \delta_{-k'}^k (u_{k}^2 u^*_{k-q} v_{k-q}
- v^*_{k+q} u_{k+q}v^2_{k}) \bra{s}\hat\gamma^\dagger_{k0}\hat\gamma^\dagger_{k1}\ket{\rm BCS}
\nonumber
\\ 
&- u_{k+q} u_{-(k'-q)}v_{-k'} v_k
\bra{s}\hat\gamma^\dagger_{(k+q)0}\gamma^\dagger_{-(k'-q)1}\hat\gamma^\dagger_{-k'0}\hat\gamma^\dagger_{k1}\ket{\rm BCS}.
\nonumber 
\end{align}
Notice that in Eq.~\eqref{eq:red} only double and quadruple 
excitations of the vacuum show up in that expression. 
Single excitations of the vacuum, for instance, contribute 
nothing in first-order perturbation theory. With this result 
we can prove that 
\begin{thm}
$\Gamma(\lambda)$ is idempotent at first order.
\end{thm}
\begin{proof}
Let us take the first-order correction to the quasiparticle vacuum:
\begin{align}
\label{eq:pert}
\ket{\Psi_1(\lambda)} &= \ket{\rm BCS} + \lambda \sum_k b_k\hat\gamma^\dagger_{k0}
\hat\gamma^\dagger_{k1}\ket{\rm BCS} + \lambda \sum_{qkk'} d_{qkk'}
\hat\gamma^\dagger_{(k+q)0}\gamma^\dagger_{-(k'-q)1}\hat\gamma^\dagger_{-k'0}
\hat\gamma^\dagger_{k1}\ket{\rm BCS} + \mathcal{O}(\lambda^2)
\end{align}
That this is the correct wave function in first order can be seen from
Eq.~\eqref{eq:red}. Notice also that in first order the wave function \eqref{eq:pert}
is correctly normalized. The parameters $d_{qkk'}$ and $b_k$ stem 
from perturbation theory. Writing the (diagonal) densities as 
\begin{align}
\rho_{kk}(\lambda) \equiv \rho^{(0)}_{kk} + \lambda \rho^{(1)}_{kk} + 
\mathcal{O}(\lambda^2) \qquad {\rm and} \qquad 
\chi_{kk}(\lambda) \equiv \chi^{(0)}_{kk} + 
\lambda \chi^{(1)}_{kk} + \mathcal{O}(\lambda^2),
\end{align}
we have for $\ket{\Psi_1(\lambda)}$
\begin{align}
\rho^{(1)}_{kk} &\equiv b_k \bra{\rm BCS}c^\dagger_{k\uparrow} c_{k\uparrow}
\gamma^\dagger_{k0}\gamma^\dagger_{k1}\ket{\rm BCS} + \text{c.c.}
= b_k u_k^*v_k^* + b_k^* u_k v_k, \\ 
\chi^{(1)}_{kk} &\equiv b_k \bra{\rm BCS}c_{k\uparrow} c_{-k\downarrow}
\gamma^\dagger_{k0}\gamma^\dagger_{k1}\ket{\rm BCS} +
 b_k^* \bra{\rm BCS}\gamma_{k1}\gamma_{k0}c^\dagger_{k\uparrow} 
c^\dagger_{-k\downarrow}\ket{\rm BCS}
 = -b_k u_k^{*2} + b^*_k v_k^2.
\end{align}
Notice that to compute the first-order contribution to the densities 
only double excitations are included. This comes from the fact that the 
expressions \eqref{eq:ap1} and \eqref{eq:ap2} only contains pairs of 
creation/annihilation operators. We also have
\begin{align*}
 2 \rho^{(0)}_{kk} \rho^{(1)}_{kk} + (\chi^{0*}_{kk}\chi_{kk}^{(1)} + \text{c.c.}) &=
 2 |v_k|^2 (b_k u_k^* v_k^* + b_k^* u_k v_k) + 
[-v_k^* u_k (-b_ku_k^{*2} + b^*_k v_k^{2}) + \text{c.c.}] \nonumber
\\
&= b_k ( |v_k|^2 u^*_k v^*_k + v_k^* u^*_k |u_k|^2) + \text{c.c.}
\nonumber \\
& = b_k u^*_k v^*_k + \text{c.c.} = \rho_{kk}^{(1)}.
\end{align*}
This proves that the idempotency holds in first order.
\end{proof}

\begin{thm}
$\Gamma(\lambda)$ is not idempotent at second order.
\end{thm}
\begin{proof}
Consider the wave function:
\begin{align}
\label{eq:sec}
\ket{\Psi_2(\lambda)} &= \ket{\Psi_1(\lambda)} +
\lambda^2 \alpha \ket{\rm BCS} +  \cdots,
\end{align}
where $\alpha$ is, once again, a parameter stemming from perturbation 
theory. Notice that
\begin{align}
\bra{\Psi_2(\lambda)}c^\dagger_{k\uparrow} c_{k\uparrow} 
\ket{\Psi_2(\lambda)} &= \rho_{kk}^{(0)} + \lambda \rho_{kk}^{(1)} + 
\lambda^2 \Big[|u_k|^2 |b_k|^2 +  |v_k|^2 \sum_{k'\neq k} |b_{k'}|^2 +
 |v_k|^2(\alpha + \alpha^*)\Big] \nonumber
\\ &+ \lambda^2 \Big[|u_k|^2 f^{(1)}_k + u_k v_k f^{(2)}_k + 
u_k^*v_k^* f^{(3)}_k + |v_k|^2 f^{(4)}_k \Big].
\end{align}
To alleviate the notation we define the following expectation values:
$$
f^{(1)}_k =   \bra{\Phi}  \hat \gamma^\dagger_{k0}\hat \gamma_{k0} \ket{\Phi}, \quad 
f^{(2)}_k =   \bra{\Phi}  \hat \gamma^\dagger_{k0}\hat \gamma^\dagger_{k1} \ket{\Phi}, \quad 
f^{(3)}_k =   \bra{\Phi}  \hat \gamma_{k1}\hat \gamma_{k0} \ket{\Phi} \quad {\rm and}
\quad f^{(4)}_k =   \bra{\Phi}  \hat \gamma_{k1}\hat \gamma^\dagger_{k0} \ket{\Phi},
$$
with
$
\ket{\Phi} = \sum_{qkk'} d_{qkk'}
\hat\gamma^\dagger_{(k+q)0}\gamma^\dagger_{-(k'-q)1}\hat\gamma^\dagger_{-k'0}
\hat\gamma^\dagger_{k1}\ket{\rm BCS}
$.
In the same vein we arrive at 
\begin{align}
\bra{\Psi_2(\lambda)}c_{k\uparrow} c_{-k\downarrow} \ket{\Psi_2(\lambda)} 
&= \chi_{kk}^{(0)} + \lambda \chi_{kk}^{(1)} + u^*_kv_k 
\lambda^2 \Big(|b_k|^2 - \sum_{k'\neq k} |b_{k'}|^2 - 
\alpha - \alpha^*\Big) \nonumber \\
&+ \lambda^2 \Big(u_k^*v_k f^{(1)}_k - u_k^{*2} f^{(3)}_k
+ v_k^2 f^{(2)}_k - v_ku_k^* f^{(4)}_k\Big).
\end{align}

Since the wave function \eqref{eq:sec} is this time non-normalized, 
normalization has to be imposed. Therefore, to second order we 
can write the (diagonal) densities as:
\begin{align}
\rho_{kk}(\lambda) &\equiv (\rho^{(0)}_{kk} + \lambda \rho^{(1)}_{kk} + \lambda^2 \rho^{(2)}_{kk} + \cdots)(1-\lambda^2 A) \\
\chi_{kk}(\lambda) &\equiv (\chi^{(0)}_{kk} + \lambda \chi^{(1)}_{kk} + \lambda^2 \chi^{(2)}_{kk} + \cdots)(1-\lambda^2 A),
\end{align}
where $A$ is the normalization factor, such that
$\bra{\Psi_2(\lambda)}\Psi_2(\lambda)\rangle = 1 + \lambda^2 A 
+ \mathcal{O}(\lambda^3)$.
Thus, $A = \sum_k |b_k|^2 + \sum_{qkk'}|d_{qkk'}|^2 + \alpha + \alpha^*$.
Notice now that second-order idempotency would imply:
\begin{align}
2 \rho^{(0)}_{kk} \rho^{(2)}_{kk} &+ (\rho^{(1)}_{kk})^2 + (\chi^{(0)*}_{kk}
\chi_{kk}^{(2)} + \text{c.c.}) + |\chi_{kk}^{(1)}|^2 - 2A  [(\rho^{(0)}_{kk})^2 
+|\chi^{(0)}_{kk}|^2] = \rho^{(2)}_{kk} - A  \rho^{(0)}_{kk}.
\end{align}
Carrying out this calculation leads to the following contradiction:
\begin{align*}
2 \rho^{(0)}_{kk} \rho^{(2)}_{kk} &+ (\rho^{(1)}_{kk})^2 + (\chi^{(0)*}_{kk}
\chi_{kk}^{(2)} + \text{c.c.}) + |\chi_{kk}^{(1)}|^2 - 2A  [(\rho^{(0)}_{kk})^2 
+|\chi^{(0)}_{kk}|^2] \\ &= 2|v_k|^2  \Big[|u_k|^2|b_k|^2 + |v_k|^2\sum_{k'\neq k} 
|b_{k'}|^2 + |v_k|^2 (\alpha + \alpha^*)\Big]
+ (b_k u_k^*v_k^* + b_k^* u_k v_k)^2 \\&+
2 |v_k|^2\Big[|u_k|^2 f^{(1)}_k + u_k v_k f^{(2)}_k + u_k^*v_k^* f^{(3)}_k +  
 |v_k|^2 f^{(4)}_k\Big] \\& - 2 |u_k|^2 |v_k|^2 \Big[|b_k|^2 - 
\sum_{k'\neq k}|b_{k'}|^2 -(\alpha + \alpha^*)\Big]
+ |b_k u_k^{*2} - b^*_k v_k^2|^2 - 2A  \rho^{(0)}_{kk} \\
&-u_k v_k^*  \Big(u_k^*v_k f^{(1)}_k - u_k^{*2} f^{(3)}_k + v_k^2 f^{(2)}_k - 
v_ku_k^* f^{(4)}_k\Big) -u^*_k v_k  \Big(u_kv_k^* f^{(1)}_k - u_k^{2} f^{(3*)}_k + v_k^{*2} f^{(2*)}_k 
- v^*_ku_k f^{(4)}_k\Big)  \\
&= 2 |v_k|^2\sum_{k'\neq k} |b_{k'}|^2 + 2 |v_k|^2 (\alpha + \alpha^*) + |b_k|^2 
- 2A  \rho^{(0)}_{kk} + u_kv_k f^{(2)}_k + u^*_kv_k^* f^{(3)}_k + 2 |v_k|^2 f^{(4)}_k
\\
&=  |v_k|^2\sum_{k'\neq k} |b_{k'}|^2 +|u_k|^2 |b_k|^2 + 2 |v_k|^2 (\alpha + \alpha^*)+
|v_k|^2\sum_{k} |b_{k}|^2 - A\rho^{(0)}_{kk}  \\ &- \Big(\sum_k |b_k|^2 + f^{(1)}_k + 
f^{(4)}_k + \alpha + \alpha^*\Big)  \rho^{(0)}_{kk} + u_kv_k f^{(2)}_k + u^*_kv_k^* f^{(3)}_k + 2 |v_k|^2 f^{(4)}_k
\\
&=  |v_k|^2\sum_{k'\neq k} |b_{k'}|^2 +  |u_k|^2 |b_k|^2 +  |v_k|^2 (\alpha + \alpha^*) - A\rho^{(0)}_{kk}  
- f^{(1)}_k  \rho^{(0)}_{kk} + u_kv_k f^{(2)}_k + u^*_kv_k^* f^{(3)}_k +  |v_k|^2  f^{(4)}_k
\\
&=   \rho^{(2)}_{kk} - A  \rho^{(0)}_{kk} - f^{(1)}_k  \neq \rho^{(2)}_{kk} - A  \rho^{(0)}_{kk}.
\end{align*}
In the second step we have used zero-order idempotency.
\end{proof}

In the above discussion we have only computed the effect of one part of the electronic 
interaction. Indeed, we are still the repulsion of two electrons with parallel spin in 
Eq.~\eqref{eq:eeint}. Yet including such a term does not change our conclusions. 
It can be easily seen from the expression for parallel-spin electronic interaction 
which is given by
\begin{align}
& \hat c^\dagger_{(k+q)\uparrow} \hat c^\dagger_{(k'-q)\uparrow}
\hat c_{k'\uparrow}\hat c_{k\uparrow} 
\\ &\qquad(u_{k+q}\hat\gamma^\dagger_{(k+q)0} + v^*_{k+q}\hat\gamma_{(k+q)1})
(u_{k'-q}\hat\gamma^\dagger_{(k'-q)0} + v^*_{k'-q}\hat\gamma_{(k'-q)1})
(u^*_{k'}\hat\gamma_{k'0} + v_{k'}\hat\gamma^\dagger_{k'1})
(u^*_{k}\hat\gamma_{k0} + v_{k}\hat\gamma^\dagger_{k1}) .
\nonumber
\end{align}
For the homogeneous electron gas we have $q \neq 0$. Therefore, 
\begin{align}
\label{eq:red2}
&\bra{s}c^\dagger_{(k+q)\uparrow} c^\dagger_{(k'-q)\uparrow}c_{k'\uparrow}c_{k\uparrow}
\ket{\rm BCS} = u_{k+q} u_{k'-q}v_{k'} v_k
\bra{s}
\gamma^\dagger_{(k+q)0}\gamma^\dagger_{(k'-q)0}
\gamma^\dagger_{k'1}\gamma^\dagger_{k1}\ket{\rm BCS}. 
\end{align}
Hence, this term only contributes to the quadruple excitations in first-order 
perturbation theory.

\end{document}